\documentclass[11pt]{article}
\usepackage[dvips]{graphics,color}
\usepackage{latexsym}
\usepackage{amsmath}
\usepackage{amssymb}
\usepackage{amsfonts,amsthm}

\newtheorem{theorem}{Theorem}

\begin{document}

\title{A Note on Fault Tolerant Reachability for Directed Graphs}

\date{\today}

\author{Loukas Georgiadis$^{1}$ \and Robert E. Tarjan$^{2}$}

\maketitle 

\footnotetext[1]{Department of Computer Science \& Engineering, University of Ioannina, Greece. E-mail: \texttt{loukas@cs.uoi.gr}.}
\footnotetext[2]{Department of Computer Science, Princeton
University, 35 Olden Street, Princeton, NJ, 08540, and Intertrust Technologies. E-mail: \texttt{ret@cs.princeton.edu}.}

In this note we describe an application of low-high orders~\cite{DomCert:TALG} in fault-tolerant network design.
Baswana et al.~\cite{FaultTolerantReachability} study the following reachability problem.
We are given a flow graph $G = (V, A)$ with start vertex $s$, and a spanning tree $T =(V, A_T)$ rooted at $s$.
We call a set of arcs $A'$ \emph{valid} if the subgraph $G' = (V, A_T \cup  A')$ of $G$ has the same dominators as $G$.
The goal is to find a valid set of minimum size.
Baswana et al.~\cite{FaultTolerantReachability} show that there is a valid set containing at most $n - 1$ arcs, and they give an algorithm to compute a minimum-size valid set in $O(m \log{n})$ time, where $n = |V|$ and $m = |A|$.
Their algorithm is based on generalizing the notion of semi-dominators~\cite{domin:lt} from a depth-first spanning tree to an arbitrary spanning tree.

Here we give a simple algorithm, Algorithm 9 below, to compute a minimum-size valid set in $O(m)$ time, given the dominator tree $D$ and a low-high order of it.
Since $D$ and a low-high order can be computed in $O(m)$ time~\cite{DomCert:TALG}, so can a minimum-size valid set.

\vspace{0.25cm}
\begin{figure}[h]
\begin{center}
\fbox{
\begin{minipage}[h]{\textwidth}
\begin{center}
\textbf{Algorithm 9: Construction of a minimum-size valid set $A'$}
\end{center}
Initialize $A'$ to be empty.
For each vertex $v \not= s$, apply the appropriate one of the following cases, where $t(v)$ and $d(v)$ are the parent of $v$ in $T$ and $D$ respectively ($d(v)$ is the immediate dominator of $v$):
\begin{description}\setlength{\leftmargin}{10pt}
\item[Case 1:] $t(v) = d(v)$.  Do nothing.
\item[Case 2:] $t(v) \not= d(v)$ and $(d(v), v) \in A$. Add $(d(v),v)$ to $A'$.
\item[Case 3:] $(d(v), v) \not\in A$.
\begin{description}\setlength{\leftmargin}{10pt}
    \item[Subcase 3a:] $t(v) > v$. Add to $A'$ an arc $(x, v)$ with $x < v$.
    \item[Subcase 3b:] $t(v) < v$. Add to $A'$ an arc $(x, v)$ with $x > v$ and $x$ not a descendant of $v$ in $D$.
\end{description}
\end{description}
\end{minipage}
}
\end{center}
\end{figure}
%\vspace{-0.5cm}

\begin{theorem}
\label{theorem:fault-tolerant-reachability}
The set $A'$ computed by Algorithm 9 is a minimum-size valid set.
\end{theorem}
\begin{proof}
First, we show that in each occurrence of Cases 2 and 3, an arc will be added to $A'$.  This is obvious for Case 2.  To show that this is true for Case 3, we use the definition of a low-high order
\cite[Section 2]{DomCert:TALG}: for any vertex $v \not=  s$ such that $(d(v), v) \not\in A$, there are two arcs $(u, v)$ and $(w, v)$ such that $u < v < w$ and $w$ is not a descendant of $v$ in $D$.  Only one of these arcs can be in $T$. It follows that in both Case 3a and Case 3b, there is an arc that satisfies the stated constraint, and hence such an arc will be added to $A'$.

Second, we show that $A'$ is valid. To do this, we use the notion of a pair of divergent spanning trees from \cite{DomCert:TALG}.  Two spanning trees $B$ and $R$, both rooted at $s$, are \emph{divergent} if, for every vertex $v$, the paths from $s$ to $v$ in $B$ and $R$ have only the dominators of $v$ in common.  Algorithm $1'$ below, which  simplifies Algorithm 1 in \cite{DomCert:TALG} as discussed just before Theorem 2.10 in \cite{DomCert:TALG}, constructs two divergent spanning trees, provided that each step is successful: the proof of Theorem 2.8 in \cite{DomCert:TALG} applies to Algorithm $1'$ equally as well as to the original Algorithm 1.
Indeed, as does Algorithm 1, Algorithm $1'$ constructs a pair of \emph{strongly} divergent spanning trees, but here we need only the weaker property of divergence.

\vspace{0.25cm}
\begin{figure}[h]
\begin{center}
\fbox{
\begin{minipage}[h]{\textwidth}
\begin{center}
\textbf{Algorithm $\mathbf{1'}$: Construction of Two Divergent Spanning Trees $B$ and $R$}
\end{center}
Let $D$ be the dominator tree of flow graph $G = (V, A)$, with a preorder that is a low-high order of $G$.
For each vertex $v \not= s$, apply the appropriate one of the following cases to choose arcs $(b(v), v)$ in $B$ and $(r(v), v)$ in $R$:
\begin{description}\setlength{\leftmargin}{10pt}
\item[Case 1:] $(d(v), v) \in A$. Set $b(v) = r(v) = t(v)$.
\item[Case 2:] $(d(v), v) \not\in A$. Choose two arcs $(u, v)$ and $(w, v)$ such that $u < v < w$ in low-high order and $w$ is not a descendant of $v$ in $D$.
Set $b(v) = u$ and $r(v) = w$.
\end{description}
\end{minipage}
}
\end{center}
\end{figure}
%\vspace{-0.5cm}

Suppose we apply Algorithm $1'$ to $G$, but only allow it to choose arcs that are in subgraph $G'$.
Suppose Case 1 of Algorithm $1'$ applies to $v$.  Then $(d(v), v)$ is in $G$. But the construction of $A'$ guarantees that $(d(v), v)$ is in $G'$ as well.
Hence this case will add an arc in $G'$, namely $(d(v), v)$, to both $B$ and $R$. Suppose Case 2 of Algorithm $1'$ applies to $v$.  Then $(d(v), v)$ is not in $G$. Algorithm 9 will apply Case 3 to $v$, which guarantees that $G'$ contains arcs $(u, v)$ and $(w, v)$ such that $u < v < w$ (in low-high order) and $w$ is not a descendant of $v$ in $D$: if $t(v) > v$, $t(v)$ cannot be a descendant of $v$ in $D$ since there is a $v$-avoiding path in $T$ from $s$ to $t(v)$. Hence this case can successfully choose arcs in $G'$ to add to $B$ and $R$. We conclude that $G'$ contains two spanning trees that are divergent in $G$. It follows that if $v$ dominates $w$ in $G'$, $v$ dominates $w$ in $G$.  The converse is immediate, since $G'$ is a subgraph of $G$. Thus $A'$ is valid.

Finally, we show that $A'$ is minimum-size. Set $A'$ contains an arc $(u, v)$ only if $t(v) \not=d(v)$. Any valid set must contain an arc entering $v$, since otherwise $t(v)$ dominates $v$ in $G'$, contradicting validity.
\end{proof}

\end{document}